\documentclass[a4paper,UKenglish,cleveref, autoref, thm-restate]{lipics-v2021}

\nolinenumbers
\title{Towards Optimal Prior-Free  Permissionless Rebate Mechanisms, with  applications to Automated Market Makers \& Combinatorial Orderflow Auctions.}
\titlerunning{Towards Optimal Prior-Free Permissionless Rebates} 

\author{Bruno Mazorra Roig}{Universtitat Pompeu Fabra, Spain}{brunomazorra@gmail.com}{https://orcid.org/0000-0003-0779-0765}{}

\usepackage{tikz-cd}
\usepackage[nocompress]{cite}
\usepackage{graphicx}

\usepackage[skins,breakable]{tcolorbox}
\graphicspath{ {./images/} }
\usepackage{float}
\newtcolorbox{mybox}[1][]{enhanced jigsaw,breakable,pad at break=1mm,
  oversize,left=8mm,interior hidden,colframe=black,nobeforeafter=,#1}

\newtcolorbox{mybox2}[2][]{%
  attach boxed title to top center
               = {yshift=-8pt},
  fonttitle    = \bfseries,
  title        = #2,#1,
  enhanced,
}

\newcommand{\x}{\textbf{x}}
\newcommand{\p}{\textbf{p}}

\author{Nicol\'as Della Penna}{Amurado Research}{nikete@gmail.com}{[orcid]}{}

\authorrunning{B. Mazorra and N. Della Penna} 

\Copyright{Bruno Mazorra} 


\category{} 

\relatedversion{} 

\acknowledgements{We want to thank Ethereum Foundation for funding Bruno's research.}

\EventEditors{}
\EventNoEds{2}
\EventLongTitle{ATF (CVIT 2023)}
\EventShortTitle{CVIT 2023}
\EventAcronym{CVIT}
\EventYear{2023}
\EventDate{June 17, 2023}
\EventLocation{}
\EventLogo{}
\SeriesVolume{42}
\ArticleNo{23}

\EventEditors{}
\EventNoEds{2}
\EventLongTitle{AFT (CVIT 2023)}
\EventShortTitle{AFT 2023}
\EventAcronym{CVIT}
\EventYear{2023}
\EventDate{June, 2023}
\EventLocation{Little Whinging, United Kingdom}
\EventLogo{}
\SeriesVolume{42}
\ArticleNo{23}

\begin{document}

\maketitle
\begin{abstract}
Maximal Extractable Value (MEV) has become a critical issue for blockchain ecosystems, as it enables validators or block proposers to extract value by ordering, including or censoring users' transactions. This paper aims to present a formal approach for determining the appropriate compensation for users whose transactions are executed in bundles, as opposed to individually. We explore the impact of MEV on users, discuss the Shapley value as a solution for fair compensation, and delve into the mechanisms of MEV rebates and auctions as a means to undermine the power of the block producer.


\end{abstract} 

\section{Introduction}

In the design of decentralized permissionless systems, it can be  desirable to rebate users part of the value they create for the system.
Applying recent advances in the understanding on the limits of permissionless false-name proof  (aka Sybil resistant) mechanisms \cite{mazorra2023cost}, we study the fundamental limits on what optimal rebates can be in such settings. We make concrete contributions to the understanding of rebate applications related to (1) automated market maker liquidity providers who have a positive externality to the network to the degree that there are liquidity premia in the system), see section \ref{section:LPrebates} (2) order flow auctions that attempt to pay back the MEV that transactions generate, and the impossibility of doing so under sybils even when it is possible without them, see section \ref{section:OrderFlow}.

Maximal Extractable Value (MEV) \cite{daian2020flash} poses significant risks to both users and developers of current blockchain ecosystems. MEV refers to the rents that validators or block proposers can extract by manipulating users' transactions, such as reordering or censoring them, often to the detriment of the users involved. A common manifestation of MEV occurs when users provide liquidity to Constant Function Market Makers (CFMM) due to adverse selection, or when users trade with CFMM through front running and sandwich attacks.

We first consider rebates to liquidity providers (LP), see section \ref{section:exampleLP}.
LPs offer complementary actions that contribute value to the overall system, while potentially suffering from adverse selection. By providing liquidity, they generate value by creating a liquidity premium across the entire market, rather than just for their own shares.

To address the issue of MEV extraction in AMM transactions, various solutions have been proposed, which can be broadly classified into three categories: preserving transaction privacy, auctioning the right to execute user transactions or sets of transactions, and ordering transactions based on their timestamps.
Some proposals, such as \href{https://penumbra.zone/}{Penumbra}, advocate for the use of cryptographic primitives to obscure the information exposed by transactions. Others propose employing Frequent Batch Auctions to mitigate one-block sandwich attacks. 
Another class of solutions is Order Flow Auctions (OFAs). They  allow users to benefit from the value they generate by facilitating an auction for the right to execute an order, with the winning bid amount given to users. The concept of OFA, also known as Payment For Order Flow (PFOF), is a long-standing practice in the traditional finance sector (TradFi), with a history dating back four decades. In crypto, an example would be auction the right to execute a swap transaction that generates a back-running arbitrage opportunity is auctioned to arbitrageurs, and the highest payment is paid back to the user responsible for generating the swap transaction. However, in general, transactions generate more revenue together than separately (i.e. the transactions complement each other) leaving an efficient outcome by auctioning the right of executing the transactions separately. For example, two transactions $A$ and $B$ provide liquidity in two (empty) CFMM with assets $T_1$ and $T_2$. We assume that the token $T_2$ is not listed in any other market maker. If after providing liquidity both CFMM have some price discrepancy, arbitrageurs can extract risk-free profit by atomically buying $T_2$ in one CFMM and selling it in the other one. However, both transactions separately do not generate any MEV opportunity. When not just the right of executing one transaction, but a set of transactions  is auctioned. While calculating the appropriate compensation for users selling their transactions individually is relatively straightforward, determining the fair compensation for bundled transactions remains a complex and unresolved issue. Our goal is to develop a methodology for ascertaining this compensation from first principles. We explore the natural equivalent of MEV-share combinatorial order flow auction \cite{MEVshare} but for welfare instead of profit maximization.



\subsection{Related work}
Since the first work about MEV was by P.Daian et al. \cite{daian2020flash}, miners, searchers and builders have extracted more than $1.1$ American billion dollars. The type of strategies cover a wield range of strategies. From MEV cyclic and Cross-domain arbitrage\cite{angeris2022constant},  sandwich and frontrunning \cite{kulkarni2022towards,wang2022impact},  multi-block MEV opportunities \cite{jensen2023multi}, to Salmonella honeypots \cite{salmonella,mazorra2022not}. In the past, the MEV has been mostly extracted by searchers and validators, however, new protocols such as Proposal-Builder-Separation \cite{PBS,PEPC}, MEV-share \cite{MEVshare}, \href{https://docs.cow.fi/}{Cow-Protocol}, \href{https://github.com/penumbra-zone/penumbra}{Penumbra} and \href{https://skip.money/docs}{Skip} present as tools to give more leverage to the users, leading to more welfare by providing a better execution. However, the complexity of providing efficient rebates is to have approximate welfare maximizing truthful (and so Sybil-proof in pseudo-anonymous systems) combinatorial auctions.

More generally, in game theory and auction theory, Sybil's strategies are usually noted as false-name strategies or shill bids. The first author studying false-name strategies in internet auctions is made in \cite{yokoo2001robust,yokoo2004effect}.
In \cite{yokoo2001robust} the authors present a combinatorial auction protocol that is robust against false-name bids. In \cite{yokoo2004effect}, M. Yokoo et. al. prove that Vickrey–Clarke–Groves (VCG) mechanism, which is strategy-proof and Pareto efficient when there exists no false-name bid, is not false name-proof/Sybil-proof and there exists no false-name proof combinatorial auction protocol that satisfies Pareto efficiency. 
In \cite{iwasaki2010worst}, the authors analyzed the worst-case efficiency ratio of false-name-proof combinatorial auction mechanisms. The authors show that the worst-case efficiency ratio of any false-name-proof mechanism that satisfies some minor assumptions is at most $2/(m+1)$ for auctions with $m$ different goods.
In cite \cite{sher2012optimal}, the author formulates the problem of optimal shill bidding for a bidder who knows the aggregate bid of her opponents. The author proves that if the mechanism is required to be Strategy-proof and false-name proof, then VCG only works with additive valuations.

On the other hand, different mechanisms have been proposed for sharing the surplus generated by players co-producing some value. For example, in \cite{balkanski2015mechanisms} consider model
procurement, where the goal is to optimize a buyer’s utility while paying sellers in a manner that reflects their contribution to the buyer’s utility. In machine learning data marketplaces, different papers such as \cite{o2015shapley,xu2021validation,nguyen2022trade} consider Shapley value as a fair mechanism for payments. This mechanism could (theoretically) be translated to MEV-sharing\footnote{As mentioned by Xinyuan Sun in \url{https://hackmd.io/@sxysun/shapley} and \url{https://github.com/flashbots/mev-research/blob/main/FRPs/active/FRP-30.md}} mechanisms by considering the buyers' utility as the MEV potentially extracted by the builder. While these mechanisms are not credible in general \cite{akbarpour2020credible}, the existence of cryptographic commitments makes it possible to create credible and strategy-proof auctions \cite{chitra2023credible} and so envision potential credible strategy-proof combinatorial OFA.

\subsection{Our contributions}

This paper examines Miner Extractable Value (MEV) rebates, and how they can be effectively managed and manipulated within the landscape of CFMM liquidity providers.

\begin{itemize}
    \item \textbf{Formalization of MEV Rebates} in context of Liquidity providers in Constant function market makers.
    \item \textbf{Desiderata \& the Sybil-Collusion-Deficit Trilemma}: propose a set of properties that would be desirable in an  MEV rebate mechanism: efficiency, symmetry, Sybil-proofness,  marginality, no deficit, Pareto-optimality, and separability.  We show that no mechanism can fulfill all of them. 
    \item \textbf{Limits of Shapley Value Rebate}: We prove that any symmetric, efficient, and Strong monotonic rebate mechanism (Shapley value mechanisms) is weak against Sybil strategies. We present an algorithm for computing the optimal Sybil strategy against this mechanism.
    \item  \textbf{Limitation of CFMM-rebates}: We show some intrinsic limitations of symmetric Sybil-proof mechanism for liquidity provision. We prove that in the prior-free setting, the value that can be rebated in any symmetric and Sybil-proof rebate mechanism decreases exponentially. We show that the natural operator that holds strong monotonicity, symmetry, additivity, and Sybil-proofness runs a deficit and therefore not practical. We present a solution for certain cooperative utility games in the Bayesian setting.
    \item \textbf{Prior-Free Optimal Rebate}: We give an operator that holds Strong monotonicity, Sybil-Proofness, Symmetry, No-deficit, and Separability that is approximately max-min optimal. 
    \item \textbf{Max-min auction rebate}: We construct an auction mechanism that is approximate worst-case optimal welfare under private valuations when there are no conflicting bundles.
\end{itemize}
\section{Example: Liquidity providers rebates}\label{section:exampleLP}
A Constant Function Market Maker (CFMM) is a type of decentralized exchange (DEX) algorithm that uses a predetermined, constant equation to automatically set the price of a token. The primary advantage of CFMMs is that they can provide liquidity to a market and determine prices without needing many buyers and sellers to achieve market efficiency. In a CFMM, the price of a token aligns closely with the off-chain price, due to the activities of arbitrageurs. These actors have the ability to exploit price differences in order to make essentially risk-free profit.

Cross-domain arbitrage refers to the exploitation of price differences between different markets, trading platforms or Blockchains. In this context, we consider a situation where there is a reference price given by a trusted auctioneer.  More formally, given a constant function market maker $\mathcal C = (\varphi,R)$, with reference price $p$, the arbitrage problem can be written as:
\begin{equation*}
\begin{aligned}
& \underset{\Delta}{\text{maximize }}p\cdot \Delta\\
& \text{subject to } \varphi(R+\Lambda_1-\Delta_1)\geq \varphi(R), \\
&\quad\quad\quad\quad\quad p\Lambda_2 = p\Delta_2,\\
&\quad\quad\quad\quad\quad \Delta_1+\Delta_2\geq \Lambda_1+\Lambda_2, \\ 
&\quad\quad\quad\quad\quad \Delta = \Delta_1+\Delta_2.
\end{aligned}
\end{equation*}
For price $p$ and reserve $R$ we can denote $\text{Arb}(p,R)$ the solution of the arbitrage optimization problem. In this context, if the arbitrageur (the builder) is welfare maximizing, a pro-rata rebate system can then be established, where liquidity providers who have taken the risk of providing liquidity are proportionally rewarded based on the profit that arbitrageurs can generate. That is, if the value of the profit is $v$ and the shares of the pool are $s\in\Delta^n$, then, the mechanism pays $p_i=s_iv$ to the player $i$.  For a given vector of shares $(s_1,...,s_n)\in\Delta^n$, reserves $R_1$ and $R_2$, $\text{Arb}$ defines a cooperative games $v:2^{[n]}\rightarrow \mathbb R_{\geq0}$ as $v(S)=\text{Arb}((\sum_{i\in S}s_i)\cdot R,p)$. Since $v$ is a cooperative game, we could consider the rebate to be the Shapley value $\phi$, however, in this context both rebates coincide.

\begin{proposition}\label{prop:CFMM} If the trading function of the CFMM, $\varphi$ is $1-$homogeneous, then $v$ is an additive game. In particular, the Shapley value rebate for liquidity providers coincide with the pro-rata rebate.
\end{proposition}

This approach can be generalized well when we have a set of CFMM and all tokens are listed in a market maker with infinite liquidity and reference price. However, a challenge arises when the tokens are not listed in an off-chain market maker, such as in the case of single domain arbitrage such as cyclic arbitrage. 
For fairness, we can use the Shapley Value, a concept from cooperative game theory. In this context, it could be used to fairly distribute the profits (or rebates) among the players (or traders) based on their individual contributions to the total profit.

However, this mechanism could face a significant problem: Liquidity Providers (LPs) have incentives to perform a "Sybil attack" on the rebate mechanism. Essentially, they can generate fake tokens and identities to manipulate the system. For example, assuming that there are three different CFMM with tokens A-C, A-B and C-B with three different liquidity providers $i=1,2,3$. If there is an arbitrage opportunity $v$, the Shapley rebate will pay each LP $v/3$ to each player, since there is no arbitrage opportunity without one of the three edges. If the LP of A-B removes the liquidity of the pool and creates two new pools with A-D and D-B, with two new identities, the Shapley rebate will pay $v/4$ to $i=2,3$ and $v/2$ to $i=1$, increasing its profits.
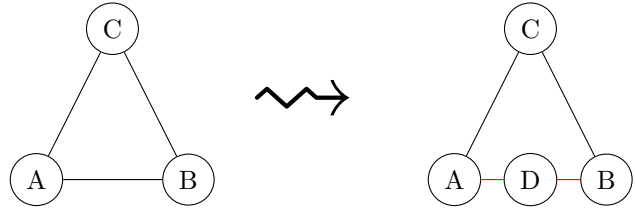
\begin{figure}[!h]
    \centering
    \begin{tikzpicture}
        \begin{scope}
            \node[draw, circle] (A) at (0, 0) {A};
            \node[draw, circle] (B) at (2, 0) {B};
            \node[draw, circle] (C) at (1, 2) {C};

            \draw (A) -- (B);
            \draw (B) -- (C);
            \draw (C) -- (A);
        \end{scope}

        \node at (3.5, 1) {\scalebox{4}[4]{$\leadsto$}};

        \begin{scope}[xshift=5.5cm]
            \node[draw, circle] (A) at (0, 0) {A};
            \node[draw, circle] (B) at (2, 0) {B};
            \node[draw, circle] (C) at (1, 2) {C};
            \node[draw, circle] (D) at (1, 0) {D}; 

            \draw [red] (A) -- (D);
            \draw[red] (D) -- (B); 
            \draw (B) -- (C);
            \draw (C) -- (A);
        \end{scope}
    \end{tikzpicture}
    \caption{The liquidity provider of the CFMM A-B creates a token D and adds two CFMM A-D and D-B.}
    \label{fig:my_label}
\end{figure}
In such a scenario, this rogue player can cannibalize the arbitrage opportunities of other players. This means they can strategically take a larger share of the available profits, undermining the fairness properties of the Shapley Value-based distribution system. In other words, the Shapley rebate mechanism is not robust in this setting. And also, this player would reduce the efficiency of the system by splitting its liquidity across multiple DEXs, due increasing of gas costs to trade from A to B.

To address this issue, in this paper, propose a solution, aiming to prevent such strategic behaviour and ensure a fair, robust, and efficient system for all participants in the CFMM environment. The solution will need to create strong disincentives for false-name and Sybil strategies, while preserving the attractiveness of the system for honest participants.
\section{MEV Preliminary}\label{section:preliminar}
In this paper, we will use the theoretical framework of \cite{babel2021clockwork,mazorra2022price}. In both paper, they introduce a similar formula to define the maximal extractable value. We will denote $\mathcal D$ as domain with the same definition as \cite{obadia2021unity}, i.e.  $\mathcal D$ is a self-contained system with a globally shared state $\texttt{st}$. This state is altered by
various agents through actions (sending transactions, constructing blocks, slashing, etc.), that execute within a shared execution environment’s semantics. With an abuse of notation, we will sometimes denote by $\mathcal D$ a union of domains. This will be useful to consider other forms of MEV, such as cross-domain arbitrage.

\begin{definition} A set of transactions $\mathcal T=\{\texttt{tx}_1,...,\texttt{tx}_k\}$ and a state $\texttt{st}$ induce an MEV opportunity in a domain $\mathcal D$ to a player $P$ if they can construct an ordered set of transactions $B$ such that:
\begin{equation}
    \Delta u_P(\texttt{st}\circ B,\texttt{st}) := u_P(\texttt{st}\circ B)-u(P,\texttt{st}) > 0,
\end{equation}
where $u_P$ is the utility function of player $P$. We call $B$ a \textit{profitable bundle} or \textit{bundle}. If $B$ consists of a unique transaction, we say that $B$ is a MEV\textit{-transaction}. For a given state $\texttt{st}$, each bundle incurs execution costs called gas costs $\textbf{g}(B)$. In general, we will assume that player that players have a linear utility in terms of the balance of a fixed asset that we will call numéraire, leading to the following definition that we will use for MEV.
\end{definition}

\begin{definition} Let $\mathcal D$ be a domain with state $\texttt{st}$, a player $P$ with local mempool view $\mathcal T^M_P$ and a set of transactions $\mathcal T_P$ that the player $P$ can construct. We denote by $\mathcal C_P=\mathcal T^M_P\cup \mathcal T_P$ to be the set of reachable transactions. We define the \textit{local} MEV \textit{of $P$ with state $\texttt{st}$ } ($\text{MEV}_P(\texttt{st})$) as the solution to the following optimization problem
\begin{align*}
& \underset{B}{\text{max}} \: \Delta b(P;\texttt{st}\circ  B,\texttt{st}) \\
& \text{s.t. } B\subseteq \mathcal C_P \\
& \text{and } \texttt{st}\rightarrow\texttt{st}\circ  B\text{ is a valid state transition in $\mathcal D$}
\end{align*}

Let $\text{argmev}_P(\texttt{st})$ be the set of bundles that are a solution to the optimization problem. The constraints of reachable bundles are subject to a player's information, gas efficiency, budget, ability to propose blocks, etc.
\end{definition}

As stated, maximal extractable value (MEV) represents the highest potential profit that can be obtained by strategically reordering, including, censoring transactions  or exploiting inefficiencies within blockchain ecosystems. Various forms of MEV exist, including internal arbitrage between constant function market makers (CFMMs), cross-domain arbitrage, and sandwich attacks. Internal arbitrage within CFMMs occurs when traders capitalize on price discrepancies between different CFMM-based decentralized exchanges (DEXs), leveraging the constant product formula to generate profit. Cross-domain arbitrage entails exploiting price differences across distinct blockchain networks, layer-2 solutions or centralized exchanges, as assets can exhibit divergent valuations in different environments. Sandwich attacks involve an informed and sophisticated actors strategically placing trades before and after a target transaction within a single block, effectively manipulating prices to their advantage and extracting value from the victim's transaction.

In general, Maximal Extractable Value rebates or MEVR for short, are a type of incentive paid by builders or sequencers to users that relay transactions that generate some surplus that the builder can extract. MEV rebates are paid out of the value generated by the transaction, and the amount of the rebate is typically determined by the MEV extracted by the builder and its paying policy. We will discuss desirable axioms of this paying policies and show the limit of these in users welfare. 

While the previous definition is able to capture these ideas, we have to extend it to incorporate the marginal contribution of players that have changed the state in previous blocks and contribute the MEV extracted. An example is liquidity providers of a constant function market makers. In this case, if a CFMM leaves a CEX-DEX arbitrage the marginal contribution of its liquidity to the Maximal extractable value, is at least, the value of this arbitrage opportunity. To capture this we will use 

To formalize MEV distribution and MEV rebates, we could use the notation introduced in \cite{mazorra2022price}. However, we can have more general definition. Let $\mathcal S$ be the set of states $\mathcal T$ the set of transactions or identities that generated state transitions in previous blocks and $\mathcal B=2^{\mathcal T}$ the set of subset of $\mathcal T$.

\textbf{Assumption 1 (Perfect MEV oracle)}: For each set of transactions $\mathcal T'$  and state $\textbf{st}$ the sequencer has access to a ``perfect MEV oracle", that is, he can efficiently compute the MEV for every set of bundles. It is well known that computing the MEV is at least an NP-problem \cite{babel2021clockwork,mazorra2022price,angeris2021note}. However, in this paper, we want to upper bound the efficiency of MEV-rebates, and so this hypothesis will not have an impact on the results. We can think of this perfect MEV oracle as non-collusive searchers in perfect market competition bidding in each state and set of bundles.
So, the sequencer can compute $\text{MEV}(\texttt{st},\mathcal P)$ where $\mathcal P=\{B_i\}_{i\in I}$. We will write $\text{MEV}$ as a map $v:\mathcal S\times 2^{\mathcal B}\rightarrow\mathbb R$. 
Observe that for a given state $\texttt{st}$ we can restrict $v$ to $\texttt{st}\times 2^{\mathcal B}$ and interpret this restriction as a \textit{coalitional form game} on a finite set of players $N$.

\textbf{Assumption 2 (Completness state-MEV)}: Let $\mathcal P$ be a mempool, i.e. a set of bundles. For every monotonic map $v:2^{\mathcal P}\rightarrow \mathbb R$, there exist a valid state $\texttt{st}\in\mathcal S$ such that $\text{MEV}(\texttt{st},\cdot)_{|2^{\mathcal P}}=v(\cdot)$.
\begin{definition} A \textit{rebate value operator} or \textit{value} operator for short is a function that assigns each valuation $v:2^{N}\rightarrow \mathbb R$ to a vector of $(\varphi_1(v),...,\varphi_n(v))$ of $\mathbb R^n$. The element $\phi_i(v)$ stands for the $i$'s rebate payment.
\end{definition}
Shapley values \cite{winter2002shapley} are a mathematical concept used to allocate the total value generated by a group of players to each individual player in a cooperative game. The Shapley value of a player is based on the player's contribution to the total value generated by the group. In the context of MEV rebates, Shapley values seems a natural solution to be used to allocate the value generated by a bundle to the different participants involved in the block. However, as we have seen in \ref{section:exampleLP}, this is in general not a good idea, since the player can strategically create multiple transactions or bundles to increase their profit. We will explain this in the following section with more detail.
\section{Sybil-proof rebates}\label{section:LPrebates}

In this paper, we explore the integration of Sybil-proof or false-name proof mechanisms with Shapley values in order to evaluate the marginal contribution of an identity to a coalition in scenarios where players can report multiple identities to maximize their payoffs. The aim is to develop a robust and "fair" framework for apportioning value among agents while preventing the manipulation of the outcome through the use of multiple false identities. In this first section we will assume that players have a fixed contribution to the coalition and the set of actions is colluding with other players and splitting into different identities. In this scenario, we will explore mechanisms where players want to truthfully reveal their identity to the mechanism. We will immediately see that this is impossible to have a mechanism that is Symmetric, Sybil-proof and collusion-proof, and so, we will explore in more depth the limits of Symmetric and Sybil-proof mechanisms. In other words, in general, we will just take into account mechanisms where players can misreport the number of identities creating Sybils, but not collude.

\subsection{Shapley value rebate}

To begin, we introduce the properties and the formula of Shapley values. The Shapley value is a widely recognized solution concept in cooperative game theory, proposed by Lloyd Shapley in 1953, which fairly distributes the gains obtained from a coalition of players. It is based on four key properties: efficiency, symmetry, null player neutral, and additivity.

Given a cooperative game $(N,v)$, where $N = \{1,2,\dots, n\}$ is the set of players and $v: 2^N \rightarrow \mathbb{R}$ is the characteristic function that assigns a real value to every coalition $S \subseteq N$, the Shapley value $\phi_i(v)$ for player $i$ can be defined as follows:

\begin{equation}
\phi_i(v) = \sum_{S \subseteq N \setminus {i}} \frac{|S|!(|N|-|S|-1)!}{|N|!} \cdot [v(S \cup {i}) - v(S)]
\end{equation}

Where $|S|$ represents the cardinality of set $S$, and $|N|$ is the total number of players in the game. The Shapley value $\phi_i(v)$ represents the expected marginal contribution of player $i$ to the coalition, averaging over all possible orders in which the players can join the coalition. The Shapley value have the following properties:

\begin{itemize}
    \item \textbf{Efficiency (E)}: The sum of the Shapley values for all players is equal to the total value of the grand coalition, i.e., $\sum_{i \in N} \phi_i(v) = v(N)$. This ensures that the entire value generated by the coalition is distributed among the players without any surplus or deficit.

    \item \textbf{Symmetry (S)}: If two players $i$ and $j$ are interchangeable, i.e., for any coalition $S \subseteq N \setminus \{i, j\}$, $v(S \cup \{i\}) = v(S \cup \{j\})$, then their Shapley values are equal, i.e., $\phi_i(v) = \phi_j(v)$. This property ensures that players contributing equally to the coalition receive equal rewards.

    \item \textbf{Null player (N)}: If a player $i$ does not contribute to any coalition, i.e., for any coalition $S \subseteq N \setminus \{i\}$, $v(S \cup \{i\}) = v(S)$, then the Shapley value of that player is zero, i.e., $\phi_i(v) = 0$. This property ensures that players who do not contribute to the coalition do not receive any payoff.

    \item \textbf{Additivity (A)}: Given two cooperative games $(N,v)$ and $(N,w)$ with characteristic functions $v$ and $w$ respectively, the Shapley value of their sum is equal to the sum of their individual Shapley values, i.e., $\phi_i(v+w) = \phi_i(v) + \phi_i(w)$ for every player $i \in N$. This property allows for the separate computation of Shapley values in different subgames and their subsequent aggregation.
\end{itemize}
Additionally, as demonstrated in the seminal work by Dubey, Shapley, and Neyman \cite{dubey1975uniqueness}, the Shapley value stands out as the unique operator that simultaneously satisfies these four properties within the domain of cooperative games defined by characteristic functions $v: 2^N \rightarrow \mathbb{R}$. The uniqueness result relies on a rigorous axiomatic analysis of cooperative game solutions, establishing the Shapley value as the sole solution concept that adheres to the principles of efficiency, symmetry, null player property, and additivity across all possible cooperative games. Consequently, this underpins the significance of the Shapley value in fairly distributing the gains obtained from a coalition of players and highlights its robustness as an equitable solution for cooperative game scenarios. 

The Shapley value holds other important properties:
\begin{itemize}
    \item An operator $\varphi$ holds \textbf{Marginality (M)} if for every pair of cooperative games $v,w$ such that $v(S\cup\{i\})-v(S) = w(S\cup\{i\})-w(S)$, then $\varphi_i(v) = \varphi_i(w)$.
    \item  An operator $\varphi$ is \textbf{Strong monotonicity (SM)} if for every pair of cooperative games $v,w$ such that $v(S\cup\{i\})-v(S)\geq w(S\cup\{i\})-w(S)$, then $\varphi_i(v)\geq \varphi_i(w)$.
\end{itemize}
Observe that if an operator holds \textbf{(SM)}, then holds \textbf{(M)}.
Young \cite{chun1989new} showed proved that there is a unique strong monotonic, symmetric and efficient operator and that this is the Shapley value. Different axiomatic approaches \cite{winter2002shapley} have been proposed for characterizing the Shapley value as a solution concept for cooperative games.

\subsection{Sybil-Proof rebates}
In the context of our paper, we aim to extend the Shapley value concept to account for the presence of multiple identities (i.e., Sybil attacks or false-name manipulations) and to devise a method for fairly computing each player's marginal contribution to a coalition, while discouraging the use of multiple identities for maximizing payoffs. 
In this scenario, the payment mechanism will consist of a map $\varphi:\bigcup_{n\geq1} TU^n_m\rightarrow \mathbb R_+^\infty$, where $TU^n_m=\{v:2^{[n]}\rightarrow \mathbb R_+|\, v\text{ monotone}\}$. The set of strategies set by players will consist of reshaping the games $v$ by extending them or reducing them in order to obtain a higher revenue, modeled by the payment mechanism $\varphi$.

The \textbf{Symmetric games} are defined in the following way. Let $R\subseteq N\setminus\{\emptyset\}$, we define $w_R:$
\begin{equation*}
    w_R(S) = \begin{cases}
                1,\text{ if }R\subseteq S,\\
                0,\text{ otherwise}.
            \end{cases}
\end{equation*}
The set of symmetric games forms a base of all cooperative games, see the appendix.

Given a cooperative game $([n],v)$, some player $i\in N$ and, some natural number $k\in\mathbb N$, we will define different Sybil extensions games as a cooperative game $([n+k],\tilde{v})$. The first, and the most intuitive one, is the one where a player replicates each identity $k$ times. 
Each one of these identities contributes exactly the same to the coalitions as the original identity. More formally, let $K=[n+k]\setminus[n]$ the Sybil extension cooperative copy game is defined as
\begin{equation*}
    v_c^K(S) = \begin{cases} v((S\setminus K)\cup \{i\}),\text { if } S\cap K\not=\emptyset,\\ v(S),\text{ if } S\cap K=\emptyset.\end{cases}
\end{equation*}
Another natural Sybil extension cooperative split game is given by the extension where each identity is dummy isolated and aggregated, they act as the original identity. More formally, 
\begin{equation*}
    v_s^K(S) = \begin{cases} v((S\setminus K)\cup \{i\}),\text { if } K\cup\{i\}\subseteq S,\\ v(S\setminus (K\cup\{i\})),\text{  otherwise}.\end{cases}
\end{equation*}

A set $\mathcal A \subseteq \bigcup_{n\geq1}\{v:2^{[n]}\rightarrow\mathbb R\}$ is Sybil extension closed with respect a Sybil extension (copy or split) if $\mathcal A$ is closed under this operation. More formally, if $v\in\mathcal A$, then $v^K\in \mathcal A$. Observe that the space of monotone functions and the space of super-modular functions are closed under the split operation. Except stated otherwise, from now on we will assume that every set function $v$ is monotone.

In the context of transactions or bundles, this is equivalent to breaking one transaction or bundle into different transactions, where both transactions need to be executed in order to make MEV. This strategic behavior can be done in various ways. One way of doing so is by creating a variable \texttt{Execute} set to \texttt{False}. Then, one sets \texttt{Execute} as \texttt{True} and the other one is equivalent to the original one but checks if \texttt{Execute} is \texttt{True}. This, can clearly be generalized to an unlimited number of transactions by defining more variables. 

More abstractly, we propose a similar definition of Sybil-proof to \cite{mazorra2022price,nowak1997axiomatization}.
\begin{itemize} 
 \item An operator $\varphi$ is \textbf{2-Efficient} if for every $i$, it holds that
 \begin{equation*}
     \varphi_i(v_c^K)+\varphi_{n+1}(v_c^K) = \varphi_{i}(v).
 \end{equation*}
 for $K=\{n+1\}$.
 \item An operator $\varphi$ is \textbf{(Ex-post) Sybil-proof (SP)}\footnote{From now on we every time we write Sybil-proof we will mean ex-post Sybil-proof} if, for every player $i$, cooperative game $v$ such that $\varphi_i(v)\geq0$, and any cooperative Sybil-extension with Sybil set $K$, we have that:
\begin{equation*}
    \varphi_i(v)\geq \max\left\{\sum_{j\in K\cup\{i\}}\varphi_j(v_c^K),\sum_{j\in K\cup\{i\}}\varphi_j(v_s^K)\right\}.
\end{equation*}
 In the literature, a more general property that takes into account any type of Sybil (that is the Sybils are not necessarily Copies or Splits) is usually denoted as super-additive, see \cite{nowak1997axiomatization}. Formally, is defined as an operator $\varphi$ is \textbf{(Ex-post) General Sybil-Proof (GSP)} if for every game $v:2^N\rightarrow\mathbb R$ and set of players $K\subseteq N$
 \begin{equation*}
    \sum_{i\in K} \varphi_i(v)\leq \varphi_p(v_p)
 \end{equation*}
 where $v_p:2^{N\setminus K\cup \{p\}}\rightarrow\mathbb R$ is the reduced game defined as $v_p(S)=S$ if $p\not\in S$ and $v_p(S)=v(S\setminus\{p\}\cup K)$ if $p\in S$.
A dual version of this concept is Collusion-proofness, that is, no subset of players is better off by pretending to be one. 
 \item Formally, an operator is \textbf{(Ex-post) Collusion-proof} if for a given a game $v$ and a subset of players $S\subseteq N$, holds 
 \begin{equation*}
     \sum_{i\in S}\varphi_i(v)\geq\varphi_p(v_S).
 \end{equation*}
    where $v_S:2^{N\setminus S\cup\{p\}}\rightarrow \mathbb R$ is a set function defined as $v_S(L)=L$ if $L\subseteq N\setminus S$ and $v_S(L)=v(L\cup S)$ if $p\in L$. The element $p$ represents the coalition of the players. 
\end{itemize}
In this work, we will sacrifice this last property due to the following lemma.
\begin{lemma}[Sybil/Collusion/Symmetry Trilemma]\label{trilemma}There is no non-trivial symmetric, Collusion-proof, and General Sybil-proof rebate mechanism.
\end{lemma}
In other words, if we want a Symmetric rebate operator, we must either sacrifice Collusion-proofness or Sybil-proofness. Since in pseudo-anonymous environments is easier to Sybil-attack than collude, see \cite{mazorra2023cost}, in the following we will just study the Sybil-proof mechanisms.

The Shapley value is, in general, not Sybil-proof under copy extension. 
For example, consider the symmetric game $w_R$. The shapley value of a player $i\in R$ is $\phi_i(w_R)=1/|R|$. If a player generates $k\geq2$ copies in total (counting $i$), then total payoff  of this extended game $w_{R\cup K}$ is $\sum_{j\in K}\phi(w_{R\cup K})= k/(|R|+k)>1/|R|$. Therefore, the Shapley value is not Sybil-proof. More generally, if a game $v$ can be expressed in terms of symmetric games with positive coefficients, i.e. $v = \sum_{R\subseteq N} c_R w_R$ with $c_R\geq0$, then there is no optimal Sybil copy strategy and the supreme is realized when $k\rightarrow+\infty$. 
\begin{theorem} \label{theorem:not_sybil} The Shapley value is not Sybil-proof. In other words, there is no operator that holds efficiency, symmetry, null player, additivity, and Sybil-proofness. And also, there is no operator that holds symmetry, marginality, efficiency, and Sybil-proofness.
\end{theorem}

Therefore, by this theorem, in order to have a Sybil-proof operator, we need to sacrifice at least one of the other properties. Following similar arguments presented in \cite{mazorra2023cost}, we will construct an operator that shrinks the payoffs for each new identity reported. As stated before, if the operator is additive is determinate by its image in the elements of the base. Similar to the proposition $2.3$ of \cite{mazorra2023cost}, we have that a symmetric and null-player operator in the symmetric games holds:
 \begin{equation}\label{eq:spoptimal}
  \varphi_i(w_R) \quad\leq\quad \begin{cases}
                    \frac{1}{2^{\mid R\mid -1}},\text{ if }i\in R,\\
                    0,\text{ otherwise.}
                \end{cases}
 \end{equation}
For completeness, the proof goes as follows:
Let $r(n)=\sum_{i\in R}\varphi_i(w_R)$ be the total reward split among $n=|R|$ players (by symmetry this is well-defined since the total value just depends on the number of players). So if a player $i$ is in $R$, then the value obtained is $r(n)/n$. Since $\varphi$ is  Sybil-proof under spliting extension, we have that  $y\geq1$, $r(1+y)/(1+y)\geq 2r(2+y)/(2+y)$, since generating two Sybils is less profitable. And so, we have that
\begin{equation*}
    r(1+n)\leq \frac{r(n)}{2}\frac{n+1}{n},\text{ for }n\geq 1.
\end{equation*}
And so, recursively we deduce that 
\begin{equation*}
 r(1+n)\leq \frac{r(1)}{2^{n}}\prod_{k=1}^{n}\frac{k+1}{k} = \frac{1}{2^{n}}(n+1)
\end{equation*}
Since $r(1)=1$, we deduce that $r(n)\leq \frac{n}{2^{n-1}}$. Now, by symmetry every player in $R$ obtains the same value, and so $\varphi_i(w_R) = r(n)/n = \frac{1}{2^{n-1}}$. $\square$

If the equality holds for all $i$ and $R\not=\emptyset$, we say that the operator is \textbf{Sybil-proof-optimal (SPO)}. Note that this bound holds due to the lack of the prior-information of the auctioneer. In case where the number of players is drawn from a distribution $\mathcal D$, the optimal (ex-ante) Sybil-proof rebate for symmetric games (full symmetric complementarity), the total welfare can be improved.

 \textbf{Prior-optimal rebate mechanism with symmetric full-complementarity}: Now, let's assume that the distribution of number of players $\mathcal D$ is common knowledge. And that agents completely complement each other, i.e. they generate a symmetric game. The mechanism designer objective is to find a symmetric mechanism that maximizes the expected welfare.

Let $\mathcal D'$ be the distribution of $\mathcal D$ conditioned to $\mathcal D\geq 1$. We write $p_n = \Pr [\mathcal D'=n]$. Our objective is to find a Sybil proof symmetric mechanism that maximizes the expected sum of payments, i.e. the welfare. We want to split part of the total value $R\in\mathbb R_{\geq0}$ generated by the complementarity among the agent. We can model the distribution with a vector $X\in\mathbb R^\infty$ such that if there are a total number of $n$ players reported, each player receives $X_n$. Therefore, if there are $n$ identities (already) reported and a player reports $k$ identities more, he obtains $k\cdot X_{k+n}$. Our objective is to find $X$ such that:
\begin{enumerate}
    \item (Ex-ante Sybil-proof) All players reporting one identity is an equilibrium (Weaker than Syibil-proof), formally
    \begin{equation}
    \mathbb E_{n\sim\mathcal D'}[X_{n}]\geq \mathbb E_{n\sim\mathcal D'}[y\cdot X_{y-1+n}]\text{ for all }y\geq 1.
    \end{equation}
    \item (Budget Balance) For every reported number of identities $n$, it holds $n\cdot X_n\leq R$.
    \item (Efficiency) $X$ maximizes the expected welfare of all players in equilibrium. That is 
\begin{equation}
X^* = \text{argmax}_{X}\quad\mathbb E_{n\sim\mathcal D'}[n\cdot X_{n}]
\end{equation}
\end{enumerate}
And so, we can rewrite this as the Linear optimization problem
\begin{equation*}
\begin{aligned}
& \underset{X}{\text{maximize }}\sum_{n=1}^{+\infty}n\cdot X_n\cdot p_n\\
& \text{subject to } 0\leq X_n\leq R/n\text{ for all }n\geq1\\
& \quad\quad\quad\quad\sum_{n=1}^\infty (X_n-X_{y-1+n}\cdot y)p_{n}\geq0\text{ for all }y\geq2
\end{aligned}
\end{equation*}
If we denote  the welfare by $W_n$ with $n$ reported players, we have that $W_n =n X_n$ and so the optimization problem can be rewritten as:
\begin{equation*}
\begin{aligned}
& \underset{W}{\text{maximize }}\sum_{n=1}^{+\infty}W_n\cdot p_n\\
& \text{subject to } 0\leq W_n\leq 1\text{ for all }n\geq1\\
& \quad\quad\quad\quad\sum_{n=1}^\infty \left(W_n/n-W_{y+n-1}y/(y+n-1)\right)p_{n}\geq0\text{ for all }y\geq2
\end{aligned}
\end{equation*}
As we will show in the figure \ref{fig:prior}, the optimal prior is in general strictly higher than the prior-free optimal rebate for full complementary transactions. An easy example to show this is when $p_n=1$ for some $n$, since $X_n=1$, $X_k=0$ for $k\not=0$ is a solution of the optimization problem and has welfare $1$. \footnote{Code of the following figure is available in \url{https://github.com/BrunoMazorra/COFA.git}.}
\begin{figure}
    \centering
    \includegraphics[scale=0.4]{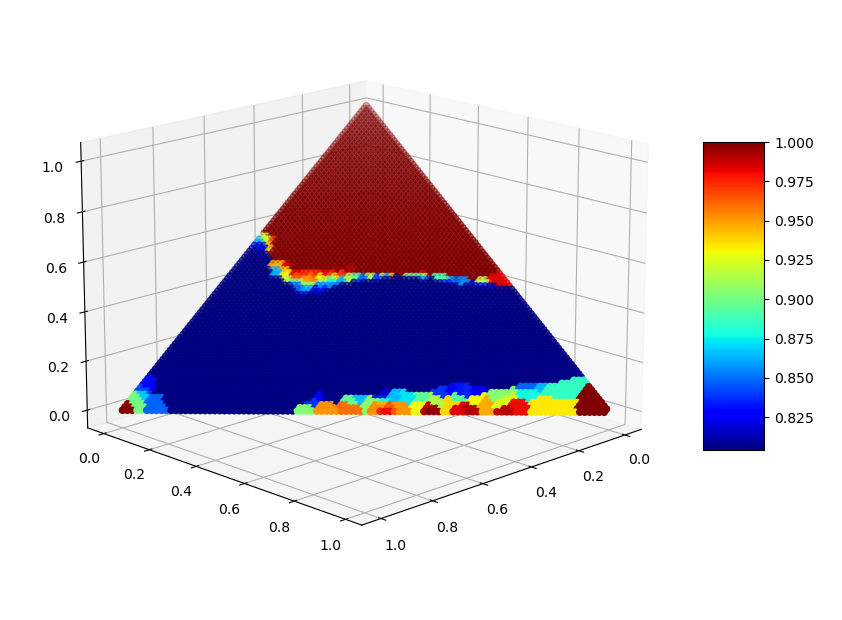}
    \caption{Plot of the solution of the optimization problem for the elements in $\{(0,0,x_3,x_4,x_5)\in\mathbb R^5_+|x_3+x_4+x_5=1\}$.}
    \label{fig:prior}
\end{figure}
Back to the prior-free setting, we have the following proposition.
 \begin{proposition}\label{prop:upperbound} Let $\hat{\varphi}$ be a symmetric and Sybil-proof operator, then $\hat{\varphi_i}(w_R)$ holds $\ref{eq:spoptimal}$. In particular, $\mathcal M^s$ consists of the symmetric and Sybil-proof operators, then \begin{equation*}
    \max_{\varphi\in \mathcal M^s}\min_{v\in TU^n_m} \frac{W(\varphi,v)}{v(N)} \leq \frac{n}{2^{n-1}}.
\end{equation*}
 \end{proposition}
If we impose that the operator is an additive, symmetric, null player, and sybil-proof optimal, then the operator must be the Banzhaf operator.
\begin{theorem}\label{theorem:banzhaf_value} There is a unique operator that is \textbf{(S)}, \textbf{(N)}, \textbf{(A)} and \textbf{(SPO)}. The operator is of the form
\begin{equation*}
    \beta_i(v)=\frac{1}{2^{\mid N\mid-1}}\sum_{S\subseteq N,i\in S}[v(S)-v(S\setminus\{i\})].
\end{equation*}
This operator is known as the \textit{Banzhaf value}. The Banzhaf value holds the marginality condition. 
\end{theorem}

We delve deeper into the properties of the Banzhaf Value, which is introduced as an extension of the classical Shapley Value that accommodates the presence of Sybil or false-name manipulations. The key properties of the Banzhaf Value is that holds marginality, symmetry, however in general, is not true that the auctioneer can make the payments without losing money. More formally,
\begin{itemize}
    \item An operator has \textbf{No-Deficit (ND)} if for every set function $v:2^N\rightarrow \mathbb R_{\geq0}$ holds $\sum_{i=1}^n\varphi_i(v)\leq v(N)$.
\end{itemize}

These properties play a crucial role in understanding the behaviour and performance of the Banzhaf Value in MEV cooperative games scenarios where players may attempt to maximize their payoffs through the use of multiple false identities. By exploring these properties, we aim to establish the Banzhaf Value as a robust and fair mechanism for apportioning value among agents in the presence of potential Sybil attacks or false-name manipulations.
\begin{proposition}\label{prop:banzhaf_properties} The Banzhaf value holds the following properties:
\begin{enumerate}
    \item The sum of payoffs hold $W(\beta,v):=\sum_{i\in N}\beta_i(v) = \frac{1}{2^{n-1}}\sum_{S\subseteq N}(2|S|- n)v(S)$.
    \item Let $\textbf{X}$ be a random variable with uniform distribution in $\{S\subseteq N:\{i\}\subseteq S\}$, then $\beta_i(v) = \mathbb E[D_i(\textbf{X})]$. In particular, if $v$ is monotone, then $\beta_i(v)\geq0$ for all $i\in N$.
    \item If $v$ is additive, then $\beta_i(v)=\phi_i(v)=v(\{i\})$.
    \item If $N=\{1,2\}$, then $\beta=\phi$.
    \item The Banzhaf value is strong monotonic.
    \item \textbf{Deficit in general}: Is not true that $\sum_i \beta_i(v)\leq v(N)$ for all cooperative games $v$. For example, let $v:\{1,2,3\}\rightarrow \mathbb R$ defined as:
    \begin{center}
        \usetikzlibrary{fadings}

        \begin{tikzpicture}
        \fill[blue] (90:3) -- (210:3) -- (-30:3) -- cycle;
        \fill[red,path fading=west] (90:3) -- (210:3) -- (-30:3) -- cycle;
        \fill[red,path fading=south] (90:3) -- (210:3) -- (-30:3) -- cycle;
        
        \node at (90:3) [above] {1};
        \node at (210:3) [below left] {1};
        \node at (-30:3) [below right] {1};
        
        \node at (145:1.5) [above left] {4};
        \node at (35:1.5) [above right] {2};
        \node at (-90:1.5) [below] {4};
        
        \node at (0,0) {5};
        \end{tikzpicture}
    \end{center}
    \begin{align*}
        v(\{1\})=v(\{2\})=v(\{3\})=1 ,\quad v(\{1,3\})=2,\quad v(\{2,3\})=v(\{1,2\})=4,\quad v(\{1,2,3\}) = 5
    \end{align*}
    In this case, we have that $W_\beta(v)=5.5>5=v(\{1,2,3\})$.
    \end{enumerate}
    
\end{proposition}

The last property have unfortunate implications. In other words, the last statement says that if we want an operator that is symmetric, additive, marginality and Sybil-proofness then in some circumstances the auctioneer must subsidize the players that form the coalition.

\begin{corollary}[Imposibility] No non-trivial operator holds the symmetry, no-deficit condition, marginality, and 2-EF. 
\end{corollary}
The previous result is deduced in \cite{nowak1997axiomatization}, where is proven that the Banzhaf value is the unique operator that holds $\textbf{(S), (2-EF)}$ and \textbf{(M)} and the fact that $\beta$ does not hold the no-deficit condition. 

\subsection{Max-min strong monotonic operator}
In proposition \ref{prop:upperbound} we have seen that if an operator is $\textbf{(S)}$ and  $\textbf{(SP)}$ then the worst-case ration is upper bounded by $\frac{n}{2^{n-1}}$. Now, we will see that this bound is approximately tight for operators that hold \textbf{(S)}, \textbf{(M)}, \textbf{(SPO)}, and \textbf{(ND)}. To do so, we will consider the \textbf{approx-Max-min optimal operator} as 
\begin{equation*}
 \psi_i(v) = \frac{1}{2^{N-1}}\sup\left\{\sum_{i\in K}\phi_k(\bar{v})\mid \bar{v}\in \text{SE}\left(v\right)\right\}
\end{equation*}
 where $\phi$ is the Shapley value and $\text{SE}(v)=\bigcup_{l\geq N}\{\bar{v}:2^{[l]}\rightarrow\mathbb R_+|\, \bar{v}\text{ general Sybil-extension of } v\}$.
 
\textbf{Observation:} For symmetric games $w_R$,  $\psi_i(v) =\frac{v(N)}{2^{n-1}}$. This holds since 
\begin{equation*}
   v(N)\geq 2^{N-1}\psi_i(v) \geq\underset{k\geq0}{\sup}\sum_{i\in K}\phi_{k}(v^K_c)=\underset{k\geq0}{\sup}\frac{k}{n+k}v(N)=v(N).
\end{equation*}
\begin{theorem}\label{theorem:main} $\psi$ holds \textbf{(S)}, \textbf{(M)}, \textbf{(SP)} and \textbf{(N-D)}. Moreover, is worst-case optimal i.e.
\begin{equation*}
    \max_{\varphi\in \mathcal M}\min_{v\in TU^n_m} \frac{W(\varphi,v)}{v(N)} \geq \min_{v\in TU^n_m} \frac{W(\psi,v)}{v(N)} = \frac{1}{2^{n-1}}
\end{equation*}
where $\mathcal M$ denotes the set of all operators that hold \textbf{(S)}, \textbf{(M)}, \textbf{(SPO)} and \textbf{(ND)}, $TU^n_m$ denotes the set of monotone cooperative games with $n$ elements and $W(\varphi,v) = \sum_{i\in N}\varphi_i(v)$. In particular, by proposition \ref{prop:upperbound}, $n\min_{v\in TU^n_m} \frac{W(\psi,v)}{v(N)} \geq \max_{\varphi\in \mathcal M}\min_{v\in TU^n_m}\frac{W(\varphi,v)}{v(N)}$.
\end{theorem}
In other words, the $\psi$ operator is an approximate worst-case optimal operator in the set of operators that hold \textbf{(S)}, \textbf{(M)}, \textbf{(SPO)} and \textbf{(ND)}. However, as we will see, this operator punishes every cooperative game $v$ to holds the "fairness" properties. Moreover, this operator punishes players that generate MEV that is separable from other players MEV.
\begin{itemize}
    \item An operator $\varphi$ is \textbf{$\alpha$-Separable ($\alpha$-SE)} if for every player $i=1,...,n$ holds $\varphi_i(v)\geq\alpha\min_{S\subseteq N,i\not\in S}\{v(S)-v(S\setminus \{i\})\}$. If we denote the left side of the inequality as the operator $\theta_i(v) =\min_{S\subseteq N\setminus\{i\}}\{v(S)-v(S\setminus \{i\})\}$, is equivalent to say $\varphi_i(v)\geq \alpha\theta_i(v)$ for all $v$ and $i=1,...,n$. As shown in \cite{nowak1997axiomatization}, this operator holds \textbf{(SM)}, \textbf{(S)}, \textbf{(ND)} and \textbf{(SP)}.
\end{itemize}
\begin{proposition}\label{prop:bound} For every monotone cooperative game, $\sum_{i=1}^n \psi_i(v)\leq \frac{nv(N)}{2^{n-1}}$. The bound is tight for symmetric games. Moreover, the operator $\psi$ is not $1-$separable. However, the operator $\bar{\psi}=\frac{1}{1+\frac{n}{2^{n-1}}}\max\{\psi,\theta\}$ is in $\mathcal M$ and is $\left(\frac{1}{1+\frac{n}{2^{n-1}}}\right)-$Separable.
\end{proposition}
This arises the question if there are other operators in $\mathcal M$ that, in general increase the welfare of the players and are separable. In other words, are there any other operators that are Pareto-optimal? 
\begin{itemize}
    \item For a given set of operators $\mathcal M$, an operator $\varphi$ is \textbf{Pareto-Optimal} if for every other operator $\varphi'$ holds $W(\varphi,v)\geq W(\varphi',v)$.
\end{itemize}
We conjecture that the operator $\psi$ is not Pareto-optimal in $\mathcal M$. Confirming or disproving this conjecture and finding better operators is a research line we aim to tackle in our future work.

\section{MEV basic Auction model}\label{section:OrderFlow}
We consider a simple model. We assume that players will have private valuations on their transactions/bundles and will be binary (the transaction has some value if it does not revert and zero otherwise). In other words, we will assume that the utility function are quasi-linear and players have private valuations. This model is clearly not sufficiently general and does not take into account the non-commutativity of transactions executions and that players valuations are clearly non-binary. However, the model presented is a sub-case of the general case, and so, the limits of the mechanisms on this model will be shared by the mechanisms for the general model.

More formally, we assume that there are an unknown number of players $n$ with quasi-linear utility $u_i$ and that have a  valuation $v_i$ for accessing to the block space. More precisely, we model the utility of the players as $u_i(x_i,v_i,p_i) = x_i v_i-p_i$ where $x_i\in\{0,1\}$ (has or does not have access to change the state) and $p_i\in\mathbb R$ is the payment (can be positive due to the potential of rebates). For a given set of transactions, we consider $\mathcal F\subseteq 2^N$ the set of feasible allocations (due to conflicting bundles). This set holds:
\begin{itemize}
    \item $\emptyset \in\mathcal F$.
    \item If $I\in \mathcal F$ then any subset $I'\subseteq I$ holds $I'\in\mathcal F$.
\end{itemize}
If the builder receives a set of bundles $2^N$, we say that a vector $\textbf{x}$ of allocation is feasible if $\{i:x_i=1\}\in\mathcal F$.



A mechanism $(\x,\p)$ is \textit{individually rational for the buyers}  if for every buyer $i$ with bundle $B_i$ and valuation $v_i$  reports $(B_i,v_i)$ to the mechanism, then $u_i(x_i,b_i,p_i)\geq0$. That is, if the player truly reports his valuation, then they are not worse than reporting nothing to the mechanism. Similarly, a mechanism $(\x,\p)$ is \textit{individually rational for the builder}  if the mechanism has no deficit. That is, a mechanism $(\x,\p)$ is \textit{or no-deficit} if the total payment made by the buyers is greater than or equal to the sum of the individual payments received by the builder. More formally, for every report $(B_1,b_1),...,(B_n,b_n)$ with outcome $a=\textbf{x}(b)$ holds $\sum_{i=1}^n p_i(a)\leq p_b(a)$.

The welfare of an allocation is then $W(\textbf{x}(b),\textbf{p}(b))= \sum_{i=1}^n u_i(x_i,v_i,p_i)$. The revenue of an allocation is $R(\textbf{x}(b),\textbf{p}(b)) = v(S)+\sum_{i=1}^n p_i(b)$. A mechanism is $\alpha-$welfare approximated in $TU_m=\{v:2^N\rightarrow \mathbb R|v\text{ monotone}\}$ if for every $v\in TU_m$, we have that 
\begin{equation*}
    W(\textbf{x}(b),\textbf{p}(b))\geq\alpha \max_{S\in\mathcal F}\left\{\sum_{i\in S}b_i+v(S)\right\}.
\end{equation*}
A mechanism $(\x,\p)$ is \textit{monotone} if for every agent $i$ and vector of reports $\textbf{b}$, we have that if $\x_i(\textbf{b})=1$, then $\x_i(\bar{b_i},\textbf{b}_{-i})=1$ for every $\bar{b_i}\geq \textbf{b}_i$. 
This simply means that if an agent $i$ is selected in the outcome by declaring a bid $\textbf{b}_i$, then by declaring a higher bid he should still be selected. Myerson’s lemma below implies that monotone algorithms admit truthful payment schemes.
\begin{theorem}[Myerson's lemma, see \cite{myerson1981optimal}]\label{theorem:myerson} In single parameter domains a normalized mechanism $\mathcal M = (\x, \p)$ is truthful if and only if:
\begin{itemize}
    \item   $\x$ is \textit{monotone}: For all $i=1,...,n$, if $b_i'\geq b_i$ and $\x(b_i,b_{-i})=1$ implies $\x(b'_i,b_{-i})=1$.
    \item  \textit{Winners are paid threshold payments}: payment to each winning bidder is $\text{inf}\{b_i| \x(b_i,b_{-i})=1\text{ and }b_i\geq0\}$.
\end{itemize}
\end{theorem}
The Myerson lemma has a harsh implication for MEV-rebate mechanisms. In some cases, every truthful mechanism can not rebate the users by the MEV that their transactions generate. The idea is the following, assume that there are two agents $i=1,2$ with valuations $v_1$ and $v_2=v_1-\varepsilon$. Also assume that both transactions conflict and both transactions generate an MEV of $v(\{1\})=v(\{2\})$ independently. Then,  by the Myerson lemma, the payment of the first player will be $v_2-\varepsilon$ independently of the value $v(\{1\})$. In other words, when there is competition for changing the same states, the rebate decreases.

Now, using the Myerson lemma, we can easily define a mechanism that is incentive compatible, and truthful that maximizes the welfare (without taking into account the surplus generated by the MEV).
\begin{mybox2}{\texttt{Myerson mechanism}}
    \begin{enumerate}
    \item Accept the pairs $(B_i,b_i)$ of bundles and bid for each player $i=1,\ldots,n$.
    \item Compute $\mathcal S_{\max}:=\underset{S\in\mathcal F}{\text{argmax}}\{\sum_{i\in S}b_i\}$.
    \item Choose randomly $S^\star\in \mathcal F_{\max}$ as an allocation.
    \item We define the payments as $p_i = \text{inf}\{b_i| \x(b_i,b_{-i})=1\text{ and }b_i\geq0\}$.
\end{enumerate}
\end{mybox2}
And the MEV maximizing mechanism:
\begin{mybox2}{\texttt{MEV maximizing mechanism}}
    \begin{enumerate}
    \item Accept the pairs $(B_i,b_i)$ of bundles and bid for each player $i=1,\ldots,n$.
    \item Compute $\mathcal S_{\max}:=\underset{S\in\mathcal F}{\text{argmax}}\{\sum_{i\in S}b_i+v(S)\}$.
    \item Choose randomly $S^\star\in \mathcal F_{\max}$ as an allocation.
    \item We define the payments as $p_i = \text{inf}\{b_i| \x(b_i,b_{-i})=1\text{ and }b_i\geq0\}$.
\end{enumerate}
\end{mybox2}

\begin{proposition}\label{prop:comparation} Both mechanisms are incentive compatible and truthful. The MEV-maximizing mechanism maximizes the total welfare (users and builder). Both mechanisms are not Sybil-proof. In other words, there is no welfare-maximizing, individually rational, and Sybil-proof Mechanism in the MEV basic auction model.
\end{proposition}
Note that the user welfare of the Myerson mechanism is not necessarily greater or equal to the user welfare of the MEV maximizing mechanism. For example, if we consider two conflicting transactions $t_1,t_2$ that generate $v>2(b_2 -b_1)$ and $0$ MEV respectively and bids $b_1<b_2$. Then, the user welfare of the Myerson mechanism is $b_2-b_1$ and the user welfare of the MEV maximizing mechanism is $b_1+v-b_2>b_2-b_1$.

The incentive compatibility and the truthfulness are deduced by the Myerson lemma \ref{theorem:myerson}. The mechanisms however are not Sybil-proof. Consider the following example. Two players want to trade in the same direction with the same amount of tokens, and this generates the same back-running cross-domain MEV opportunity. This leads to two conflicting transactions $t_1$ and $t_2$. Players' valuations and bids on the trade are $b_1=1\$$ and $b_2=(1+\varepsilon)\$$. The MEV generated by $t_1$ is $10\$$ and the second is $10\$$ as well. Then, the second transaction is allocated with payment $p_2=1$, having a utility $u_2=\varepsilon$. If one player breaks the second transaction by using conflicting opcodes into two complementary transactions $t^1_i$ and $t^2_i$ with bids $b^1_{i}=b^2_i=1$, then by Myerson-lemma we have that the transactions $t^1_i$ and $t^2_i$ are allocated in the block but with payments $p^1_i=p^2_i=0$. And so, the mechanisms are not Sybil-proof.

Also, even without taking into account sybils, it is generally inaccurate to state that truthful, non-deficit, and individually rational mechanisms exist in block space allocation that can deliver $\alpha$-optimal welfare for any strictly positive $\alpha$. For example, when a player can have negative valuations (costs) on the allocation. We intend to illustrate this through an ensuing counterexample. The fundamental explanation for this lies in the interaction of bundles: when two bundles complement each other, any truthful mechanism will be obligated to compensate both participants according to their marginal contribution, presenting an inherent challenge to the mechanisms with no deficit.

\begin{proposition}[Negative result] If players have negative valuations, then there is no truthful and no deficit mechanism that is better than $0-$approximate welfare in the worst-case scenario.
\end{proposition}
\begin{proof} Consider the following valuation $v(\{1,2\})=1$ and $v(\{1\})=v(\{2\})=v(\emptyset)=0$ with costs $c_1=0$ and $c_2=0$.  If a mechanism is strictly better than $0-$approximate welfare, then the mechanism will output $x_1=1,x_2=1$. By  Myerson's lemma (for costs), the first player in this scenario must pay $p_1=1$ and $p_2=1$. However, $v(\{1,2\})=1$ and so has a deficit since the builder must pay $2$. So the allocation is $x_i=0$ for some $i=1,2$ leading to $0$ social welfare.
\end{proof}

In case $\mathcal F=2^N$, by the previous section, for every operator $\tau\in\mathcal M$ the following mechanism is truthful, symmetric, worst-case approximate welfare maximizing, Sybil-proof,  has no-deficit and is strong monotonic.
\begin{mybox2}{\texttt{$\tau$-mechanism}}
    \begin{enumerate}
    \item Accept the pairs $(B_i,b_i)$ of bundles and bid for each player $i=1,\ldots,n$.
    \item Take $S_{\max}=N$.
    \item Take the payments $p_i =-\tau_{i}(v)$.
\end{enumerate}
\end{mybox2}

The fundamental problem to solve in the general case is when there are  players that complement each other to share the costs (in a Sybil-proof manner) to outbid the conflicting bundles. In future work, we will explore Sybil cost-sharing mechanisms following the work of \cite{dobzinski2008shapley,gkatzelis2016optimal}.
\section{Discussion \& Open questions}
In this paper, we discussed rebates for LP of CFMM that are strongly monotonic, symmetric, non-deficit, and strategy-proof against creating false tokens. Moreover, we provided an approximate max-min operator that maintains these properties. Additionally, we formalized the maximum welfare operator with priors for players who have full complementarity. However, the task of finding Pareto-optimal operators with these properties in both the prior-free and prior-optimal settings is left for future work. Moreover, finding a computationally an efficient algorithm to estimate the value of the operator, that retains the incentive properties also remains for future work. 

Regarding MEV-sharing, we discussed a simple model with private valuations in the allocation and independent of the order of execution. Generally, though, players have interdependent valuations and exhibit preferences in the order of execution. In future work, we aim to analyze Sybil-proof mechanisms in these settings. Specifically, we postponed the study of cost-sharing Sybil-proof mechanisms that uphold the no-deficit condition.

\appendix
\begin{proof}\ref{prop:CFMM} Observe  that if $\Delta$ is a solution of the optimization problem with reserves $R$, then $\lambda\Delta$ is a solution for the problem with $\lambda R$ reserves. And so the $\text{Arb}(\lambda R,p) = \lambda \text{Arb}(R,p)$. Since $\text{Arb}$ is one dimensional, we have that $\text{Arb}(s_1+s_2,p)=\textbf{Arb}(s_1,p)+\text{Arb}(s_2,p)$, and so, $v$ is additive.
\end{proof}

\begin{proof}\ref{trilemma} Lets assume that $\psi$ is symmetric, Collusion-proof and General Sybil-Proof. Consider a symmetric game $w_N$ the Since is Collusion-proof, if $S=N$, we deduce that $\sum_{i=1}^n\varphi_i(w_N)\geq \varphi_1(v)$, where $v$ is the game with a unique player and $v(\{1\})=w_N(N)=R$. On the other hand, since the operator is Sybil-proof under copies we have that $\sum_{i=1}^n\varphi_i(w_N)\leq R\frac{n}{2^{N-1}}$. Using both bounds, we have that $\varphi_1(v)\leq R\frac{n}{2^{n-1}}$. If $\varphi_i(v)\not=0$, this leads to a contradiction, since the right side of the equation tends to $0$ when $n\rightarrow +\infty$. 
Therefore, $\varphi_1(v)=0$ for every $v:2^{[1]}\rightarrow \mathbb R_{\geq0}$. Since $\varphi$ is GSP we have that the for any $\hat{v}:2^N\rightarrow\mathbb R$ monotone such that $\hat{v}(N)=v(\{v_1\})$ holds $\sum_{i=1}^N \varphi_i(v)\leq \varphi_1(v)=0$. Therefore, $\varphi_i(v)=0$.
\end{proof}

\begin{lemma}\label{lemma:base} The set $\{w_R:R\subseteq N\}$ is a base of the real vector space $\{v:2^N\rightarrow \mathbb R\}$. Moreover, for any game $v$, it holds
\begin{equation*}
    v = \sum_{R\not=\emptyset,R\subseteq N} c_Rw_R.
\end{equation*}
The coefficients are independent of $N$, and are given by
\begin{equation*}
    c_R = \sum_{T\subseteq R}(-1)^{\mid R\mid -\mid T\mid}v(T).
\end{equation*}
\end{lemma}
In particular, if an operator $\phi$ is additive, then is completely determinate by its evaluation in the symmetric games.
\begin{proof}\ref{theorem:banzhaf_value} First, since the operator is symmetric and sybil-proof efficient, we have that it holds \ref{eq:spoptimal} with equality. Then, we deduce the uniqueness by additivity. Let's compute it. We know by lemma \ref{lemma:base} that for a cooperative game $v$, it holds
\begin{equation*}
 v = \sum_{R\not=\emptyset,R\subseteq N} \left ( \sum_{S\subseteq R}(-1)^{\mid R\mid -\mid S\mid}v(S)\right)w_R.
\end{equation*}
and so, we have that:
\begin{align*}
    \phi_i(v) &= \sum_{R\not=\emptyset,R\subseteq N} \left ( \sum_{S\subseteq R}(-1)^{\mid R\mid -\mid T\mid}v(S)\right)\phi_i(w_R)\\
    &=\sum_{R\subseteq N,i\in R} \left ( \sum_{S\subseteq R}(-1)^{\mid R\mid -\mid S\mid}v(S)\right)\frac{1}{2^{\mid R\mid-1}}\\
    &=\sum_{S\subseteq N}\sum_{R\subseteq N,S\cup\{i\}\subseteq R} (-1)^{\mid R\mid -\mid S\mid}\frac{1}{2^{\mid R\mid -1}}v(S)\\
\end{align*}
From the previous equality, let us write
\begin{equation*}
    \gamma_i(S) = \sum_{R\subseteq N,S\cup\{i\}\subseteq R} (-1)^{\mid R\mid -\mid S\mid}\frac{1}{2^{\mid R\mid -1}}
\end{equation*}
It is easy to see that if $i\not\in S$, then $\gamma_i(R)=-\gamma(R\setminus\{i\})$. All the terms in the right-hand side of the previous equation will be the same in both cases, except that there is a change of sign throughout. This means, we will have that
\begin{equation*}
    \phi_i(v)= \sum_{R\subseteq N, i\in R}\gamma_i(R)[v(R)-v(R\setminus\{i\})]
\end{equation*}
Let $s= |S|$. Now if $i\in R$, we see that there are exactly ${n-r \choose r-s}$ coalitions $R$ with $r$ elements such that $S\subseteq R$. Thus, we have that,
\begin{align*}
    \gamma_i(S) &= \sum_{r=s}^n (-1)^{r-s}{n-s \choose r-s}\frac{1}{2^{r-1}}\\
    &=\sum_{l=0}^{n-s}(-1)^l{n-s \choose l}\frac{1}{2^{l+s-1}}\\
    &=\frac{1}{2^{s-1}}\sum_{l=0}^{n-s}(-1)^l{n-s \choose l}\frac{1}{2^{l}}\\
    &=\frac{1}{2^{s-1}}(1-1/2)^{n-s} = \frac{1}{2^{n-1}}.
\end{align*}
and thus deducing that $\phi^s_i(v)=\frac{1}{2^{\mid N\mid-1}}\sum_{S\subseteq N,i\in S}[v(S)-v(S\setminus\{i\})]$. Now let's prove that is copy and split Sybil-proof.

\textbf{Copy case}: Let $i$ be a player that generates a set of $K=K'\cup\{i\}$ sybils, such that $i\not\in K'$  with $\mid K'\mid=k$. Then:
\begin{align*}
    \sum_{j\in K\cup\{i\}}\phi_j(\tilde{v}) &= \sum_{j\in K }\frac{1}{2^{n+k-1}}\sum_{S\subseteq N\cup K,j\in S}[\tilde{v}(S)-\tilde{v}(S\setminus i)]\\
    &=\frac{1}{2^{n+k-1}}\left(\sum_{j\in K }\sum_{\substack{S\subseteq N\cup K,\\ \mid S\cap K\mid\geq2}}\underbrace{[\tilde{v}(S)-\tilde{v}(S\setminus i)]}_{=0}+\sum_{j\in K }\sum_{\substack{S\subseteq N\cup K,\\ \mid S\cap K\mid=1}}\underbrace{[\tilde{v}(S)-\tilde{v}(S\setminus i)]}_{=v(S)-v(S\setminus i)}\right)\\
    &=\frac{\overbrace{\mid K\mid}^{=k+1}}{2^{n+k-1}}\sum_{S\subseteq N,i\in S}[v(S)-v(S\setminus\{i\})]\leq \frac{1}{2^{n-1}}\sum_{S\subseteq N,i\in S}[v(S)-v(S\setminus\{i\})]
\end{align*}

\textbf{Split case}: Let $i$ be a player that generates a set of $K=K'\cup\{i\}$ sybils, with $i\not\in K'$ such that $i\not\in K'$  with $\mid K'\mid=k$.
\begin{align*}
    \sum_{j\in K\cup\{i\}}\phi_j(\tilde{v}) &= \sum_{j\in K }\frac{1}{2^{n+k-1}}\sum_{S\subseteq N\cup K,j\in S}[\tilde{v}(S)-\tilde{v}(S\setminus i)]\\
    &=\frac{1}{2^{n+k-1}}\left(\sum_{j\in K }\sum_{S\subseteq N\cup K, K\subseteq S}[\underbrace{\tilde{v}(S)-\tilde{v}(S\setminus \{j\})}_{=v((S\setminus K)\cup\{i\})-v(S\setminus K)}]\right)\\
    &=\frac{\mid K\mid}{2^{n+k-1}}\sum_{S\subseteq N,i\in S}[v(S)-v(S\setminus\{i\})]\leq  \frac{1}{2^{n-1}}\sum_{S\subseteq N,i\in S}[v(S)-v(S\setminus\{i\})]
\end{align*}
This concludes the proof of the theorem.
\end{proof}
\begin{proof}\ref{prop:upperbound} Is deduced immediately by the \ref{eq:spoptimal}.
\end{proof}

\begin{proof}\ref{prop:banzhaf_properties}
\begin{itemize}
    \item This is shown in \cite{dragan1996new}.
    \item Trivial.
    \item $\frac{1}{2^{n-1}}\sum_{S\subseteq N\setminus\{i\}}[v(S\cup \{i\})-v(S)]=\frac{1}{2^{n-1}}\sum_{S\subseteq N\setminus\{i\}}v(i) = v(i)$.
    \item This is shown in \cite{nowak1997axiomatization}.
\end{itemize}
\end{proof}

\begin{proof} \ref{theorem:not_sybil} Let's do it first for the copy Sybil extension. Consider the game $v_{[n]}(L)=1$ for all $L\subseteq [n]$. Then, using the Shapley value formula, we have that $\phi_i(v_{[n]})=1/n$. The copy Sybil extension game of the player $i=1$ with one Sybil becomes $v_{[n+1]}$, and so the payoff of the player $i=1$ is $2\phi_1(v_{[n+1]})=2/(n+1)>1/n$ for $n\geq2$. For the split Sybil extension, follows analogously considering the Symmetric game $w_{[n]}$.
\end{proof}
\begin{proof} \ref{theorem:main} Let $\psi_i(v) = \frac{1}{2^{N-1}}\sup\left\{\sum_{i\in K}\phi_k(\bar{v})\mid \bar{v}\in \text{SE}\left(v\right)\right\}$. First, we will prove that $\psi$ 

\begin{itemize}
    \item \textbf{Marginality}: Let $v,w\in TU^n_m$ and $v'_i\geq w'_i$. Then, clearly $v'^{K}_j\geq w'^{K}_j$ for all $j\in [N]\cup K$. Since $\phi$ holds marginality, we have that $\phi_k(v^K)\geq\phi_k(w^K)$ and so $\psi_i(v)\geq \psi_i(w)$.
    \item \textbf{No deficit}: Let $v\in TU^n_m$. Since $v^K(N\cup K)=v(K)$, we have that $\sum_{i\in K}\phi_i(v^K)\leq v(N)$ so, $\psi_i(v)\leq \frac{v(N)}{2^{n-1}}$. Therefore, $\sum_{i=1}^n \psi_i(v)\leq \frac{nv(N)}{2^{n-1}}\leq v(N)$. 
    \item \textbf{Symmetry}: Is deduced by the symmetry of $\phi$.
    \item \textbf{Min}: Holds $\psi_i(v)\geq \frac{1}{2^{n-1}}\phi_i(v)$. Therefore $\sum_{i=1}^n\psi_i(v)\geq \frac{1}{2^{n-1}}$. So, $\min_{v\in TU^n_m}\frac{W(\psi,v)}{v(N)} \geq \frac{1}{2^{n-1}}$. The other inequality is deduced by using Proposition 2.3. \cite{mazorra2023cost} and the fact that $\psi$ is Sybil-proof.
    \item \textbf{General-Sybil-proof}: Suppose a player generates $k$ sybils (counting the original one) augmenting the game $v$ to $w$. 
    By construction for each Sybil we have that 
    \begin{equation*}
            \sup\left\{\sum_{i\in K}\phi_k(\bar{v})\mid \bar{v}\in \text{SE}\left(v\right)\right\}\geq \sup\left\{\sum_{i\in K}\phi_k(\bar{v})\mid \bar{v}\in \text{SE}\left(w\right)\right\}   
    \end{equation*}
    Therefore,
    \begin{align*}
        \frac{1}{2^{N-1}}\sup\left\{\sum_{i\in K}\phi_k(\bar{v})\mid \bar{v}\in \text{SE}\left(v\right)\right\}&\geq\frac{k}{2^{N-1+k-1}} \sup\left\{\sum_{i\in K}\phi_k(\bar{v})\mid \bar{v}\in \text{SE}\left(v\right)\right\}\\
        &\geq\sum_{k=1}^n\frac{1}{2^{N-1}}\sup\left\{\sum_{i\in K}\phi_k(\bar{v})\mid \bar{v}\in \text{SE}\left(w\right)\right\}
    \end{align*}
\end{itemize}
\end{proof}
\begin{proof}\ref{prop:bound} 
The  $\sup\left\{\sum_{i\in K}\phi_k(\bar{v})\mid \bar{v}\in \text{SE}\left(v\right)\right\}\leq v(N)$. Therefore,  $\psi_i(v) \leq \frac{v(N)}{2^{n-1}}$, we have that $\sum_{i=1}^n\psi_i(v)\leq \frac{nv(N)}{2^{n-1}}$. The operator $\bar{\psi}$ is clearly Symmetric and strong monotonic. Is Sybil-proof since both $\theta$ and $\psi$ are general spill-proof. The opeartor $\bar{\psi}$ has no-deficit by the following argument:
\begin{align*}
\left(1+\frac{n}{2^{n-1}}\right)\sum_{i=1}^n\bar{\psi}_i(v) &= \sum_{i=1}^n\max\{\psi_i(v),\theta_i(v)\}\\
                            &\leq\sum_{i=1}^n\max\{\frac{v(N)}{2^{n-1}},\theta_i(v)\}\\
                            &\leq \sum_{i=1}^n\left [\frac{v(N)}{2^{n-1}}+\theta_i(v)\right ]\\
                            &\leq \frac{nv(N)}{2^{n-1}} + v(N)=\left(1+\frac{n}{2^{n-1}}\right)v(N)
\end{align*}
\end{proof}
So $\sum_{i=1}^n\bar{\psi}_i(v)\leq v(N)$. 
Is $\left(\frac{1}{1+\frac{n}{2^{n-1}}}\right)-$Separable by construction since
\begin{equation*}
    \psi_i(v) =\left(\frac{1}{1+\frac{n}{2^{n-1}}}\right)\max\{\psi_i(v),\theta_i(v)\}\geq \left(\frac{1}{1+\frac{n}{2^{n-1}}}\right)\theta_i(v).
\end{equation*}
\begin{proof}\ref{prop:comparation} The incentive compatibility and truthfulness are deduced by the monotonicity of the allocation and the Myerson lemma.
\end{proof}

\end{document}